\documentclass[a4paper,11pt]{article}

\usepackage{natbib}

\usepackage{amssymb,type1cm}

\usepackage[pdftex,bookmarks=true,bookmarksopen=true]{hyperref}

\title{The Nakamura numbers \\ for computable simple games\thanks{%
Preprint, Social Choice and Welfare (2008) 31:621--640,
\href{http://dx.doi.org/10.1007/s00355-008-0300-5}{doi:10.1007/s00355-008-0300-5}
}}

\author{Masahiro Kumabe \\
Kanagawa Study Center, The University of the Air\\
2-31-1 Ooka, Minami-ku, Yokohama
232-0061, Japan 
\and 
H.  Reiju Mihara\thanks{Corresponding author. \protect\\
\emph{URL:} \url{http://econpapers.repec.org/RAS/pmi193.htm} (H.R. Mihara).
}\\
Graduate School of Management, Kagawa University\\
Takamatsu 760-8523, Japan}

\date{February 2008}

\newcommand{\A}{\mathcal{A}}
\newcommand{\N}{\mathbb{N}}
\newcommand{\REC}{\mathrm{REC}}
\newcommand{\CRec}{\mathrm{CRec}}

\newcommand{\p}{\mathbf{p}}
\newcommand{\pref}{\succ_i^\p}
\newcommand{\profile}{\mbox{$(\pref)_{i\in N}$}}
\newcommand{\xprefy}{\{\,i\in N:{x\pref y}\,\}}

\newcommand{\qed}{\enspace\enspace \vrule height 6pt width5pt
depth2pt}

\usepackage{theorem}

\newtheorem{theorem}{Theorem}
\newtheorem{prop}[theorem]{Proposition}

\newtheorem{lemma}[theorem]{Lemma}
\newtheorem{claim}[theorem]{Claim}

{\theorembodyfont{\normalfont}
\newtheorem{definition}{Definition}
\newtheorem{example}{Example}
\newtheorem{remark}{Remark}
}

\newenvironment{proof}{\emph{Proof}.}{\qed\bigskip}

\begin{document} 

\maketitle 

\begin{abstract} 
The Nakamura number of a simple game plays a critical role in preference aggregation (or multi-criterion ranking):
the number of alternatives that the players can always deal with rationally is less than this number.
We comprehensively study the restrictions that various properties for a simple game impose on its Nakamura number.
We find that a computable game has a finite Nakamura number greater than three only if it is proper, nonstrong, and 
nonweak, regardless of whether it is monotonic or whether it has a finite carrier.
The lack of strongness often results in alternatives that cannot be strictly ranked.

\emph{Journal of Economic Literature} Classifications:  C71, C69, D71.

\end{abstract}

\pagebreak

\section{Introduction}



The \emph{Nakamura number} plays a critical role in the study of 
preference aggregation rules with \emph{acyclic} social preferences.\footnote{
 \citet{banks95}, \citet{truchon95}, and \citet{andjiga-m00} are recent contributions to the literature.
Earlier papers on acyclic rules can be found in \citet{truchon95} and \citet{austensmith-b99}.
Note that acyclicity of a preference is necessary and sufficient for the existence of a maximal element
on every finite subset of alternatives.
When the weak social preferences are required to be \emph{transitive}, 
we are back in Arrow's difficult setting (\citeyear{arrow63}).} %
Consider a \emph{(simple) game}\footnote{Simple games are often referred to as ``voting games'' in the literature.
In this paper, we sometimes call them ``games'' for short.}---a coalitional game that assigns either 0 or~1 to each coalition:
those assigned~1 are winning coalitions and those assigned~0 are losing coalitions.
Combining the game with a set of alternatives and a profile of individual preferences, one obtains a 
\emph{simple game with (ordinal) preferences}, from which one can
derive a social preference (dominance relation).
Nakamura's theorem (\citeyear{nakamura79}) gives a necessary and sufficient condition for a simple game with
preferences to have a nonempty \emph{core} (the set of maximal elements of the social preference) for all profiles: 
the number of alternatives is less than a certain number 
(the smallest number of winning coalitions that collectively form an empty intersection), 
called the \emph{Nakamura number} of the simple game.
Thus the greater the Nakamura number for a given game is, the larger the set of alternatives is
from which the rule (mapping from profiles to social preferences) can always find a maximal element.


\citet[Theorem~17]{kumabe-m08jme} extend Nakamura's theorem to their framework
and apply it to \emph{computable} simple games.
They show that every (nonweak) computable game has a finite Nakamura number.
This implies that under the preference aggregation rule based on a computable game,
the number of alternatives that the set of players can deal with rationally is restricted by this number.
(Remark~\ref{nakamura-formal} gives a formal discussion of this result.)

We are therefore interested in the question of how large the Nakamura number can be.
In fact, \citet[Proposition~15]{kumabe-m08jme} show that every integer $k\ge 2$ is 
the Nakamura number of some computable game.
Of course, a large Nakamura number can be attained only by satisfying or violating certain properties for simple games.
For example, the Nakamura number of a nonproper game, which admits two complementing winning coalitions,
is at most~2 (Lemma~\ref{nonproper-nakamura2}).


In this paper, we study the restrictions that various properties (axioms) for a simple game impose 
on its Nakamura number.
We restrict our attention to the computable simple games and classify them into thirty-two ($2^5$) classes
 in terms of their \emph{types} (with respect to monotonicity, properness,\footnote{\label{nonproper-foundation}%
While simple games are often defined so that they are monotonic and proper, 
we allow simple games to be nonmonotonic or nonproper for completeness.
We can derive such games from a strategic game form, giving a justification (strategic foundation) 
for including them.
For example, we obtain a \emph{nonproper} game from the game form~$g$,
defined by $g(0,0)=g(0,1)=g(1,0)=0$ and $g(1,1)=1$,
which describes the unanimous voting rule.
Each player is effective for the set $\{0\}$ in the sense that by choosing~$0$,  she can force the outcome to be
in the set.
Then the simple game consisting of the coalitions that are effective for $\{0\}$ is nonproper.
For another example,
we obtain in Remark~\ref{strategic-foundation} an important class of \emph{nonmonotonic} games
from a certain class of game forms.} %
 strongness, and nonweakness)
and finiteness (existence of a finite carrier).
Table~\ref{summary} summarizes the results.  
For example, a type 5 $(+-++)$
(monotonic, \emph{non}proper, strong, nonweak)
computable game has Nakamura number equal to~2, whether it is finite or infinite.\footnote{%
Strictly speaking, we only assert in this paper that the numbers in each entry in the table 
are \emph{not ruled out}; we are not much interested in asserting
that \emph{every} entry not indicated ``none'' contains a game in which an empty coalition is losing.
However, those who accept the results in \citet{kumabe-m07csg64}
will find the latter assertion acceptable.
For most entries, the examples given in the paper cited suffice.
For the other entries, we need to modify the examples---which we do,
with the exception of a few entries (footnote~\ref{notallexamplesgiven}).} %
Note that the Nakamura number for a \emph{weak} game is infinite by definition.

\renewcommand{\thetable}{\arabic{table}}
\begin{table}
\caption{Possible Nakamura Numbers for Computable Games} %
\begin{center}
\begin{tabular}{rllcrll} \hline\hline
Types & Finite & Infinite & & Types & Finite & Infinite  \\
 \cline{1-3} \cline{5-7}
 $1 (++++)$ & 3 & 3 & & $9 (-+++)$ & 2 & 2 \\
 $2 (+++-)$ & $+\infty$  & none & & $10 (-++-)$ & none & none \\
 $3 (++-+)$ & $\ge 3$ & $\ge 3$ & & $11 (-+-+)$ & $\ge 2$ & $\ge 2$ \\
 $4 (++--)$ & $+\infty$ & $+\infty$ & & $12 (-+--)$ & $+\infty$ & $+\infty$ \\
 $5 (+-++)$ & 2 & 2  & & $13 (--++)$ & 2 & 2 \\
 $6 (+-+-)$ & none  & none & & $14 (--+-)$ & none & none \\
 $7 (+--+)$ & 2 & 2 & & $15 (---+)$ & 2 &  2 \\
 $8 (+---)$ & none  & none & & $16 (----)$ & none & none \\
\hline \\  %
\end{tabular}
\parbox{110mm}{\footnotesize
Possible Nakamura numbers are given in each entry, assuming that
an empty coalition is losing (so that the Nakamura number is
at least~2).
The types are defined by the four conventional axioms: 
monotonicity, properness, strongness, and nonweakness. For example,
the entries corresponding to Type~2 $(+++-)$ indicates that among the computable, monotonic ($+$), proper ($+$), strong ($+$), weak ($-$, because \emph{not} nonweak) games, finite ones have a Nakamura number equal to $+\infty$
and infinite ones do not exist.}  %
\end{center}\label{summary}
\end{table}

We make two observations from Table~\ref{summary}.
First, \emph{a nonweak computable game has a Nakamura number greater than~3 only if 
it is proper, and \emph{non}strong}
(i.e., either of type~3 $(++-+)$ or of type~11 $(-+-+)$).\footnote{%
Propositions~\ref{nakamura:type3} and \ref{nakamura:type3inf} state that any Nakamura number $k\ge 3$
is attainable by type 3 finite and infinite games.
Propositions~\ref{nakamura:type11} and \ref{nakamura:type11inf} state that any Nakamura number $k\ge 2$
is attainable by type 11 finite and infinite games.
Remark~\ref{strategic-foundation} gives a strategic foundations for these games.} %
In particular, for the players to be always able to choose a maximal element from at least three alternatives, 
strongness of the game must be forgone (unless the game is dictatorial (type~2)).
The reader should not overlook the importance of the number 3 in the above observation.
It is the Nakamura number of the majority game with an odd number of (at least three) players.
To deal with three or more alternatives rationally (though it is generally impossible to rank them \citep{arrow63})
requires a Nakamura number greater than~3.
Second, as far as computable games are concerned, \emph{a number~$k$ 
is the Nakamura number of a \emph{finite} game of a certain type (except type~2) if and only if 
it is that of an \emph{infinite} game of the same type}. 
Restricting games to finite ones does not reduce or increase
 the number of alternatives that the players can deal with rationally.
 
In contrast, if we drop the computability condition, these observations are no longer true.
A ``nonprincipal ultrafilter,'' which is noncomputable and has an infinite Nakamura number \citep{kumabe-m08jme},
serves as a counterexample to both:
It is a nonweak game with a Nakamura number greater than~3, but it is strong.  
It is a type~1 \emph{infinite} game with a Nakamura number different from~3, 
the Nakamura number of type~1 \emph{finite} games.
In fact, one can use ultrafilters not only
to find a maximal element from any finite set of alternatives (regardless of the size),
but also
to rank (while preserving the transitivity of the weak social preference) any number of alternatives 
\citep[Section~5]{kumabe-m08jme}.
This fact explains why nonprincipal ultrafilters are used for resolving Arrow's impossibility~(\citeyear{arrow63}).
The lack of computability of nonprincipal ultrafilters, however, implies that
such resolutions are impractical \citep{mihara97et}.

\bigskip

The rest of the Introduction gives a background briefly.
Much of it is fully discussed in \citet{kumabe-m08jme}.

One can think of simple games as representing voting methods or multi-criterion decision rules.
They have been central to the study of social choice \citep[e.g.,][]{peleg02hbscw}.
For this reason, the paper can be viewed as a contribution to the foundations of 
\emph{computability analysis of social choice}, 
which studies algorithmic properties of social decision-making.\footnote{%
This literature includes \citet{kelly88}, 
\citet{lewis88}, \citet{bartholdi-tt89vs,bartholdi-tt89cd},
\citet{mihara97et,mihara99jme,mihara04mss}, \citet{kumabe-m08jme,kumabe-m07csg64},
and \citet{tanaka07}.}  %

The importance of computability in social choice theory would be unarguable.
First, the use of the language by social choice theorists suggests the importance: 
for example, \citet{arrow63} uses words such as ``\emph{process or rule}'' or
``\emph{procedure}.'' 
Second, there is a normative reason:
computability of social choice rules formalizes the notion of ``due process.''\footnote{%
\citet{richter-w99jme} give further justifications for studying 
computability-based economic theories.}

We consider an infinite set of ``players.''
Roughly speaking, a simple game is \emph{computable} if there is a Turing program (finite algorithm) that can decide
from a description (by integer) of each coalition whether it is winning or losing.
Since each member of a coalition should be describable, we assume that the
set $N$ of (the names of) players is countable, say, $N=\N=\{0,1,2, \dots \}$.
Also, we describe coalitions by a Turing program that can decide for the name of each player whether
she is in the coalition.  Since each Turing program has its code number (G\"{o}del number),
the coalitions describable in this manner are describable by an integer, as desired.
(Such coalitions are called \emph{recursive} coalitions.)

\citet{kumabe-m08jme} give three interpretations of countably many \emph{players}:
(i)~generations of people extending into the indefinite future,
(ii) finitely many \emph{persons} facing countably many \emph{states} of the world~\citep{mihara97et}, and
(iii)~attributes or \emph{criteria} in multi-criterion decision-making.\footnote{
Legal decisions involve~(iii).  \citet{kumabe-m07csg64} discuss the formation of legal precedents, in which
an infinite number of criteria are potentially relevant but only finitely many of them are actually cited.} %
We can naturally re-interpret the preference aggregation problem
 (which provides motivation for studying the Nakamura number)
as a \emph{multi-criterion ranking} problem, for example.
In multi-criterion ranking, each criterion ranks finitely many alternatives;
we are interested in aggregating those countably many rankings into one (acyclic relation). 
Assuming that the underlying simple game is computable is intuitively plausible
in view of the following consequences:
(i)~each criterion is treated differently;\footnote{%
Computable simple games violate anonymity \citep[Proposition~13]{kumabe-m08jme}.} %
(ii)~whether an alternative has a higher rank than another can be determined by
examine finitely many criteria, 
though how many criteria need to be examined depends on each situation (Proposition~\ref{cutprop}).
The (lack of strongness) observation mentioned above suggests that rational choice from many (at least three)
 alternatives often involves alternatives that cannot be strictly ranked.

\section{Framework} %

\subsection{Simple games}\label{notions}

Let $N=\N=\{0,1,2, \dots \}$ be a countable set of (the names of) 
players.  Any \textbf{recursive} (algorithmically decidable) 
subset of~$N$ is called a \textbf{(recursive) coalition}.

Intuitively, a simple game describes in a crude manner the power 
distribution among \emph{observable} (or describable) coalitions (subsets of players). 
We assume that only \textbf{recursive} coalitions are observable.  
According to \emph{Church's thesis}~\citep{soare87,odifreddi92}, the recursive coalitions 
are the sets of players for which there is an algorithm that can decide for 
the name of each player whether she is in the set.\footnote{\citet{soare87} and \citet{odifreddi92}
give a more precise definition of \emph{recursive sets} as well as detailed discussion of recursion theory.
The papers by \citet{mihara97et,mihara99jme} contain short reviews of recursion theory.} 
Note that \textbf{the class~$\REC$ of recursive 
coalitions} forms a \textbf{Boolean algebra}; that is, it includes $N$ 
and is closed under union, intersection, and complementation.

Formally, a \textbf{(simple) game} is a collection~$\omega\subseteq\REC$ of (recursive) coalitions.
We will be explicit when we require that $N\in \omega$.
The coalitions in $\omega$ are said to be \textbf{winning}.  
The coalitions not in $\omega$ are said to be \textbf{losing}. 
One can regard a simple game as a function from~REC to $\{0,1\}$, assigning the value 1 or 0 to each 
coalition, depending on whether it is winning or losing.

We introduce from the theory of cooperative games a few basic 
notions of simple games~\citep{peleg02hbscw,weber94}.
A simple game $\omega$ is said to be 
\textbf{monotonic} if for all coalitions $S$ and $T$, the 
conditions $S\in \omega$ and $T\supseteq S$ imply $T\in\omega$.  
$\omega$ is \textbf{proper} if for all recursive coalitions~$S$, 
$S\in\omega$ implies $S^c:=N\setminus S\notin\omega$.  $\omega$ is 
\textbf{strong} if for all coalitions~$S$, $S\notin\omega$ 
implies $S^c\in\omega$.  $\omega$ is \textbf{weak} if 
$\omega=\emptyset$ or
the intersection~$\bigcap\omega=\bigcap_{S\in\omega}S$ of the winning coalitions is nonempty.  
The members of $\bigcap\omega$ are called \textbf{veto players}; they 
are the players that belong to all winning coalitions.  
(The set $\bigcap\omega$ of veto players may or may not be observable.)
$\omega$ is \textbf{dictatorial} if there exists some~$i_0$ (called a 
\textbf{dictator}) in~$N$ such that $\omega=\{\,S\in\REC: i_0\in 
S\,\}$.  Note that a dictator is a veto player, but a veto player is 
not necessarily a dictator.
It is immediate to prove the following well-known lemmas: 

\begin{lemma} \label{weakisproper}
If a simple game is weak, it is proper.\end{lemma}

\begin{lemma} \label{strongweakisdic}
A simple game is dictatorial if and only if it is strong and weak.\end{lemma}

A \textbf{carrier} of a simple game~$\omega$ is a coalition $S\subset N$ 
such that
\[ 
T\in\omega \iff S\cap T\in \omega
\]
for all coalitions~$T$.
When a game~$\omega$ has a carrier~$T$, we often restrict the game on $T$ and
identify $\omega$ with $\omega |T:=\{S\cap T: S\in \omega \}$.
We observe that if $S$ is a carrier, then so is any coalition $S'\supseteq S$.
Slightly abusing the word, we sometimes say a game is \textbf{finite} if it has a finite carrier; 
otherwise, the game is \textbf{infinite}.

The \textbf{Nakamura number} $\nu(\omega)$ of a game~$\omega$ is 
the size of the smallest collection of winning coalitions having empty intersection
\[
\nu(\omega)=\min \{\#\omega': \textrm{$\omega' \subseteq \omega $ and $\bigcap\omega'=\emptyset$} \}
\]
if $\bigcap\omega=\emptyset$ (i.e., $\omega$ is nonweak); 
otherwise, set $\nu(\omega)=+\infty$, which is understood to be greater than any cardinal number.
In computing the Nakamura number for a game, it suffices to look only at the subfamily of minimal winning coalitions,
\emph{provided that the game is finite}.  If the game is infinite, we cannot say so since minimal winning coalitions
 may not exist.
 
Extending and applying the well-known result by \citet{nakamura79},
\citet{kumabe-m08jme} show that
computability of a game entails a restriction on the number of alternatives
that the set of players (with the coalition structure described by the game) can deal with rationally.
The following remark gives a formal presentation of that result, adapted to the present framework.

\begin{remark} \label{nakamura-formal}%
Let $X$ be a (finite or infinite) set of \emph{alternatives}, with cardinal number $\#X\geq 2$.
Let $\A$ be the set of \emph{(strict) preferences}, i.e., 
\emph{acyclic} (for any finite set $\{x_1, x_2, \ldots, x_m\}\subseteq X$,
if $x_1 \succ x_2$, \ldots, $x_{m-1} \succ x_m$, then $x_m \not\succ x_1$;
in particular, $\succ$ is asymmetric and irreflexive) binary relations~$\succ$ on $X$. 
A \emph{profile} is a list 
$\p=\profile \in\A^{N}$ of {individual preferences} $\pref$
such that $\xprefy \in\REC$
for all~$x$, $y\in X$.

A \emph{simple game with (ordinal) preferences} is a list $(\omega, X, \p)$ of
a simple game~$\omega$ in which an empty coalition is losing,
 a set $X$ of alternatives, and a profile $\p$. %
Given a simple game with preferences, 
we define the dominance relation (social preference)~$\succ^\p_\omega$ 
by $x\succ^\p_\omega y$ if and only if there is a winning coalition
$S\in\omega$ such that $x\pref y$ for all $i\in S$.
Note that the mapping $\succ_\omega$ 
from profiles $\p$ to dominance relations $\succ^\p_\omega$ defines an \emph{aggregation rule}.
The \emph{core} $C(\omega, X, \p)$ of the simple game with preferences is the set of undominated alternatives:
\[
C(\omega, X, \p)=\{x\in X: \textrm{$\not\!\exists y\in X$ such that $y\succ^\p_\omega x$}\}.
\]
\citet[Corollary~19]{kumabe-m08jme} show that \emph{if $\omega$ is computable and nonweak,
then there exists a finite number $\nu$ (the Nakamura number~$\nu(\omega)$) such that
the core $C(\omega, X, \p)$ is nonempty for all profiles~$\p$
if and only if $\#X<\nu$}.\end{remark}

\subsection{The computability notion}
\label{comp:notion}

To define the notion of computability for simple games, we first introduce  
an indicator for them.  In order to do that, 
we first represent each recursive coalition by a characteristic index ($\Delta_0$-index).  
Here, a number $e$~is a \textbf{characteristic index} for a coalition~$S$
if $\varphi_e$ (the partial function computed by the Turing program with code number~$e$)
is the characteristic function for~$S$. 
Intuitively, a characteristic index for a coalition describes
the coalition by a Turing program that can decide its membership.
The indicator then assigns the value 0 or 1 to each 
number representing a coalition, depending on whether the 
coalition is winning or losing.  When a number does not represent a 
recursive coalition, the value is undefined.

Given a simple game $\omega$, its \textbf{$\delta$-indicator} is the partial 
function~$\delta_\omega$ on~$\N$ defined by
\[
	\delta_\omega(e)=\left\{
	    \begin{array}{ll}
		1 & \mbox{if $e$ is a characteristic index for a recursive
	set in $\omega$},  \\
		0 & \mbox{if $e$ is a characteristic index for a recursive
	set not in $\omega$},  \\
		\uparrow & \mbox{if $e$ is not a characteristic
	index for any recursive set}.
	\end{array}
	\right.
\]
Note that $\delta_\omega$ is well-defined since each $e\in\N$ can be a 
characteristic index for at most one set.

\medskip

We now introduce the notion of \emph{($\delta$)-computable} games.
We start by giving an intuition.
A number (characteristic index) representing a coalition
(equivalently, a Turing program that can decide the membership of the coalition)
is presented by an inquirer to the aggregator (planner), 
who will compute whether the coalition is winning or not.
The aggregator cannot know a priori which indices will possibly be presented to her.
So, \emph{the aggregator should be ready to give an answer whenever a characteristic index for 
some recursive set is presented to her}.
This intuition justifies the following condition of computability.\footnote{
\citet{mihara04mss} also proposes a stronger condition, \emph{$\sigma$-computability}.
We discard that condition since it is too strong a notion of computability (Proposition~3 of that paper;
for example, even \emph{dictatorial} games are not $\sigma$-computable).} %

\begin{description} \item[($\delta$)-computability] $\delta_\omega$ has 
an extension to a partial recursive function.  
\end{description}

\section{Preliminary Results} %
\label{prelim}

In this section, we give a sufficient condition and a necessary condition for a game to be computable.

\medskip

\textbf{Notation}.  We identify a natural number~$k$ with the finite 
set $\{0,1,2,\ldots,k-1\}$, which is an initial segment of~$\N$.  
Given a coalition $S\subseteq N$, we write $S\cap k$ to represent the 
coalition $\{i\in S: i<k\}$ consisting of the members of $S$ whose 
name is less than~$k$.  
We call $S\cap k$ the \textbf{$k$-initial segment of $S$}, and view it 
either as a subset of~$\N$ or as the string $S[k]$ of length~$k$ of 0's and 1's 
(representing the restriction of its characteristic function to 
$\{0,1,2,\ldots,k-1\}$).\enspace$\|$

\begin{definition}
Consider a simple game.  A string $\tau$ (of 0's and 1's) of 
length~$k\geq 0$ is \textbf{winning determining} if any 
coalition $G\in\REC$ extending $\tau$ (in the sense that $\tau$ is an 
initial segment of $G$, i.e., $G\cap k=\tau$) is winning; $\tau$ is 
 \textbf{losing determining} if any 
coalition $G\in\REC$ extending $\tau$ is losing.  
A string is \textbf{determining} if it is either winning determining or losing determining.
A string is \textbf{nondetermining} if it is not determining.\end{definition}

The following proposition restates a sufficient condition \citep[the ``if'' direction of Theorem~4]{kumabe-m08jme} 
for a game to be computable.
In particular, \emph{finite games are computable}.
The proposition can be proved easily:

\begin{prop} [{\citet{kumabe-m07csg64}}] \label{delta0det2}
Let $T_0$ and $T_1$ be recursively enumerable sets of (nonempty) strings such that
any coalition has an initial segment in $T_0$ or in $T_1$ but not both.
Let $\omega$ be the simple game defined by $S\in \omega$ if and only if
 $S$ has an initial segment in~$T_1$.
Then $T_1$ consists only of winning determining strings, $T_0$ consists only of losing determining strings, 
and $\omega$ is $\delta$-computable.\end{prop}

The following proposition \citep[Proposition~3]{kumabe-m08jme}
gives a necessary condition for a game to be computable:

\begin{prop}[{\citet{kumabe-m08jme}}] \label{cutprop}
Suppose that a $\delta$-computable simple game is given.  
\textup{(i)}~If a coalition $S$ is winning, then it has an initial segment $S[k]$ (for some $k\in \N$) 
that is winning determining.
\textup{(ii)}~If $S$ is losing, then it has an initial segment $S[k]$ that is losing determining.
\end{prop}

\section{The Main Results} 
\label{}

We classify computable games into thirty-two ($2^5$) classes as shown in~Table~\ref{summary},
in terms of their \textbf{(conventional) types} 
(with respect to the conventional axioms of monotonicity, properness, strongness, and nonweakness)
and finiteness (existence of a finite carrier).
Among the sixteen types, five (types 6, 8, 10, 14, and 16) contain no games;
also, the class of type~$2$ infinite games is empty 
(since type 2 games are dictatorial).\footnote{%
These results, also found in \citet{kumabe-m07csg64}, are immediate from Lemmas \ref{weakisproper} and~\ref{strongweakisdic}.}%

We therefore have only $(16-5)\times 2-1=21$ classes of games to be checked.
For each such class, we find the set of possible Nakamura numbers.
We do so, whenever important, by constructing a game in the class having a particular Nakamura number,
unless the example given in~\citet{kumabe-m07csg64} suffices.\footnote{%
Some examples in \citet{kumabe-m07csg64} violate the condition that $\emptyset$ is losing, 
which we impose in this paper.
In this paper, we omit examples of games with a small Nakamura number 
when the construction is based on the details of the paper cited.
Specifically, we relegate examples of a type 9 infinite game and a type 13 infinite game
to Appendix~B.%
\label{notallexamplesgiven}} %

\emph{We only consider games in which $\emptyset$ is losing}.  
Otherwise, the Nakamura number for the game becomes~1---not a very interesting case.
(Also, note that if $\emptyset$ is winning and the game has a losing coalition, then it is \emph{non}monotonic.)

We consider \emph{weak} games first.
Among the weak games, types 2, 4, and 12 are nonempty.\footnote{%
These types, being weak, consist of games in which $\emptyset$ is losing.
\citet{kumabe-m07csg64} give examples of these types of games.} %
By definition, their Nakamura number is infinite.
We have so far examined all the types whose labels are even numbers.

We henceforth consider \emph{nonweak} (hence nonempty by definition) computable games.
\citet[Corollary~16]{kumabe-m08jme} show that
they have finite Nakamura numbers:

\begin{lemma}[{\citet{kumabe-m08jme}}] \label{nakamura-finite}
Let $\omega$ be a computable, nonweak simple game.  
Then, its Nakamura number~$\nu(\omega)$ is finite.\end{lemma}

\subsection{Small Nakamura numbers} \label{section:small}

First, the definition of proper games implies the following:\footnote{%
The conditions $\omega\neq \emptyset$ in 
Lemmas~\ref{nonproper-nakamura2} and \ref{nonmon-strong-nakamura} are redundant,
since an empty game is monotonic, proper, nonstrong, and weak, according to \emph{our} definition.
We retain the conditions in parentheses,
since the definitions of these properties are not well-established for an empty game.}

\begin{lemma} \label{nonproper-nakamura2}
Let $\omega$ be a game satisfying $\emptyset \notin\omega$ (and $\omega\neq \emptyset$).
If $\omega$ is nonproper, then $\omega$ is nonweak with $\nu(\omega)=2$.\end{lemma}

Lemma~\ref{nonproper-nakamura2} is equivalent to the assertion that 
a game is proper if its Nakamura number~$\nu(\omega)$ is at least~3.
It does not rule out the possibility that proper games have Nakamura number equal to~2.
Lemma~\ref{nonproper-nakamura2} implies that
the games of types 5, 7, 13, and 15 have Nakamura number equal to~2.
Example~\ref{type13ex15ex} gives examples of type~13 and type~15 finite games.\footnote{
It is easy to show that types 5 and 7 contain games in which $\emptyset$ is losing.
If $\emptyset$ were winning, then by monotonicity the game would consist of all coalitions (a type~5 game).  
Since the examples of types 5 and 7 games in \citet{kumabe-m07csg64} all have losing coalitions,
$\emptyset$ is losing in those games.
The type 15 infinite game in that paper satisfies the condition that $\emptyset$ is losing.
To show that type 13 contains an infinite game in which $\emptyset$ is losing
is more delicate, but can be done (Appendix~B) by modifying the example in that paper.}  

\begin{example} \label{type13ex15ex}
We first give a type~$13$ finite game.
Let $T=\{0, 1,2\}$ be a carrier and let $\omega |T:=\{S\cap T: S\in \omega\}$
consist of $\{0, 1, 2\}$,  $\{1, 2\}$, $\{0\}$, $\{1\}$, $\{2\}$.
The other three coalitions in~$T$ are losing.
Then, $\omega$ is nonmonotonic, nonproper, strong, and nonweak with $\nu(\omega)=2$.

We next give a type~$15$ finite game.
Let $T=\{0, 1,2\}$ be a carrier and let $\omega |T$
consist of $\{0, 1, 2\}$,  $\{1, 2\}$, $\{0\}$, $\{1\}$.
The other four coalitions in~$T$ are losing.
Then, $\omega$ is nonmonotonic, nonproper, nonstrong, and nonweak with $\nu(\omega)=2$.\end{example}

Next, we consider computable \emph{strong} games that are nonweak.
These games have Nakamura numbers not greater than~3:

\begin{lemma} \label{strong-nakamura3}
Let $\omega$ be a computable, strong nonweak game satisfying $\emptyset \notin\omega$.
Then $\nu(\omega)=2$ or $3$.\end{lemma}

\begin{proof} %
Since $\omega$ is computable, by Proposition~\ref{cutprop}, 
every winning coalition has a finite subcoalition that is winning,
which in turn has a minimal winning subcoalition that is winning.
If there is only one minimal winning coalition~$S\neq \emptyset$, 
then the intersection of all winning coalitions is~$S$, which is nonempty;
this violates the nonweakness of ~$\omega$.  So there are at least two
(distinct) minimal winning coalitions $S_1$ and $S_2$ in~$\omega$.
Let $S=S_1\cap S_2$.
$S$ is losing since it is a proper subcoalition of the minimal winning coalition~$S_1$.
Then, since $\omega$ is strong, $S^c$ is winning.
Since $S_1\cap S_2 \cap S^c=S\cap S^c=\emptyset$, we have
$\nu(\omega)\le 3$ by the definition of the Nakamura number.
The assumption that $\emptyset \notin\omega$ rules out $\nu(\omega)=1$.
($\nu(\omega)=2$ if there are distinct minimal winning coalitions $S_1$ and $S_2$ 
such that $S=S_1\cap S_2=\emptyset$; otherwise,
$\nu(\omega)=3$.)\end{proof}

\begin{remark}
The computability condition cannot be dropped from Lemma~\ref{strong-nakamura3}
(a minimal winning coalition may not exist if a winning coalition has no finite, winning subcoalition).
A nonprincipal ultrafilter is a counterexample; it has an infinite Nakamura number.
(See \citet[Sections 2.1 and 4.3]{kumabe-m08jme} for the definition of a \emph{nonprincipal ultrafilter} and 
the observation that  it has no finite winning coalitions and is noncomputable, 
monotonic, proper, strong, and nonweak.)\end{remark}

\begin{lemma} \label{mono-prop-nakamura}
Let $\omega$ be a monotonic proper game satisfying $\emptyset \notin\omega$ and $\omega\neq \emptyset$.
Then $\nu(\omega)\ge 3$.
\end{lemma}

\begin{proof}
Suppose $\nu(\omega)=2$.
Then, there are winning coalitions $S$, $S'$ whose intersection is empty.
That is $S'\subseteq S^c$.
By monotonicity, $S^c$ is winning, implying that $\omega$ is not proper.\end{proof}

\begin{lemma} \label{nonmon-strong-nakamura}
Let $\omega$ be a nonmonotonic strong game satisfying $\emptyset \notin\omega$ (and $\omega\neq \emptyset$).
Then $\omega$ is nonweak with $\nu(\omega)= 2$.
\end{lemma}

\begin{proof}
Since nonempty $\omega$ is nonmonotonic, 
there exist a winning coalition $S$ and a losing coalition $S'$ such that
$S\cap S'^c=\emptyset$.  
This means that the Nakamura number is 2, since $S'^c$ is winning by strongness of $\omega$.\end{proof}

Lemma~\ref{strong-nakamura3} and Lemma~\ref{mono-prop-nakamura}
imply that type~$1$ games have a Nakamura number equal to~3.
Lemma~\ref{nonmon-strong-nakamura}
implies that type~$9$  games have a Nakamura number equal to~2.
Proposition~\ref{type1ex} and Example~\ref{type9ex} give examples of these games:\footnote{%
We can also give an example of an infinite, computable, type~$9$ game (Appendix~B).
It rests on the details of the construction in \citet{kumabe-m07csg64}.}

\begin{prop}\label{type1ex}
There exist finite, type~$1$ (i.e., monotonic proper strong nonweak) games and 
infinite, computable, type~$1$ games.
\end{prop}

\begin{proof}
An example of a type~$1$ finite game is the majority game with an odd number of (at least three) players.
An example of a type~$1$ infinite game is given in Appendix~\ref{nice_game1}.\end{proof}

\begin{example} \label{type9ex}
We give a type~$9$ finite game.
Let $T=\{0, 1,2\}$ be a carrier and let $\omega |T:=\{S\cap T: S\in \omega\}$
consist of $\{0, 1, 2\}$, $\{0\}$, $\{1\}$, $\{2\}$.
The other four coalitions in~$T$ are losing.
Then, $\omega$ is nonmonotonic, proper, strong, and nonweak with $\nu(\omega)=2$.\end{example}

\subsection{Large Nakamura numbers}

Having considered all the other types of games,
we now turn to types $3$ and~$11$ (i.e., proper nonstrong nonweak games).
These are the only types that may have a Nakamura number greater than~3.

First, we consider games with finite carriers.  
An example of a game having Nakamura number equal to~$k\ge 2$ can be defined 
on the carrier $T=\{0,1, \ldots, k-1\}$;
the game~$\omega$ consists of the coalitions excluding at most one player in the carrier:
$S\in \omega$ if and only if $\#(T\cap S) \geq k-1$.
We extend this example slightly:

\begin{prop}\label{nakamura:type3}
For any $k\geq 3$,  there exists a finite, computable, type~$3$ 
(i.e., monotonic proper nonstrong nonweak) game~$\omega$
with Nakamura number $\nu(\omega)=k$.
\end{prop}

\begin{proof}
Given $k\geq 2$, 
let $\{T_0,T_1, \ldots, T_{k-1}\}$ be a partition of a finite carrier $T=\bigcup_{l=0}^{k-1} T_l$.
Define $S\in \omega$ iff $\#\{T_l: T_l\subseteq S\} \geq k-1$.  
Then it is straightforward to show that $\omega$ is monotonic and nonweak with $\nu(\omega)=k$.
Now, suppose that $k\ge 3$. 
To show that $\omega$ is proper, suppose $S\in \omega$.
Then $S$ includes at least $k-1$ of the partition elements~$T_l$, implying that 
$S^c$ includes at most one of them.
To show that $\omega$ is nonstrong, suppose that a partition element, 
say $T_l$, contains at least two players, one of whom is denoted by~$t$.
We then have the following two losing coalitions complementing each other:
(i)~the union of $k-2$ partition elements~$T_{l'}$ and $\{t\}$
and (ii)~the union of the other partition element and $T_l\setminus\{t\}$.\end{proof}

\begin{remark}\label{k2}
Because of Lemma~\ref{mono-prop-nakamura}, Proposition~\ref{nakamura:type3}
precludes $k=2$.
Note that the game in the proof is nonproper if and only if $k=2$.
If $k \le 3$, then it generally fails to be strong, though it is indeed strong if
all the partition elements~$T_l$ consist of singletons.\end{remark}

\begin{prop} \label{nakamura:type11}
For any $k\geq 2$,  there exists a finite, computable, type~$11$
(i.e., nonmonotonic proper nonstrong nonweak) game~$\omega$
with Nakamura number $\nu(\omega)=k$.
\end{prop}

\begin{proof}
Given $k\geq 3$, 
let $\{T_0,T_1, \ldots, T_{k-1}\}$ be a partition of a finite carrier $T=\bigcup_{l=0}^{k-1} T_l$.
Define $S\in \omega$ iff $\#\{T_l: T_l\subseteq S\} = k-1$.  
Then $\omega$ is nonmonotonic; the rest of the proof is similar to that
of Proposition~\ref{nakamura:type3}.

For $k=2$, we give the following example:
Let $T=\{0, 1, 2\}$ be a carrier and define $\omega|T=\{S\cap T: S\in \omega\}=\{\{0\}, \{1\} \}$.
It is nonmonotonic since $\{0\}\in \omega$ but $\{0,1\}\notin \omega$.
It is proper: $S\in \omega$ implies $S\cap T=\{0\}$ or $\{1\}$,
which in turn implies $S^c\cap T=\{1,2\}$ or $\{0, 2\}$, neither of which
is in $\omega|T$; hence $S^c\notin\omega$.
It is nonstrong since $\{0, 1\}$ and $\{2\}$ are losing.
It is nonweak with $\nu(\omega)=2$ since the intersection of the winning coalitions
$\{0\}$ and $\{1\}$ is empty.\end{proof}

\begin{remark}[Strategic Foundations] 
 \label{strategic-foundation}
We justify (give a strategic foundation for) type~3 and type~11 simple games 
having Nakamura number equal to $k\geq 3$ by deriving them from certain game forms.
These types particularly deserve justification,
since they are the only types that contain (two games with different Nakamura numbers and)
games with an arbitrarily large Nakamura number.

Let $g\colon \prod \Sigma_i \to X$ be a game form on the set $\{0, 1, \ldots, k-1\}$ of players,
defined by $g(\sigma)=1$ if and only if $\#\{ i: \sigma_i=1\} \ge k-1$,
where $\Sigma_i=\{0, 1\}$ is the set of player $i$'s strategies and
$X=\{0, 1\}$ is the set of outcomes.  One can think of the game form as representing a
voting rule in which no individual has the veto power.
We claim that, depending on the notion of \emph{effectivity} employed,
the simple game derived from~$g$ is either 
(i)~the type 3 game consisting of the coalitions containing at least $k-1$ players 
(a game in the proof of Proposition~\ref{nakamura:type3}) or
(ii)~the type 11 game consisting of the coalitions made up of exactly $k-1$ players
(a game in the proof of Proposition~\ref{nakamura:type11}).

(i)~For each coalition $S\subseteq I$, let $\Sigma_S:= \prod_{i\in S} \Sigma_i$
and $\Sigma_{-S}:= \prod_{i\notin S} \Sigma_i$ be the collective strategy set of $S$
and that of the complement.
A coalition $S$ is \emph{$\alpha$-effective for} a subset $B \subseteq X$ 
if $S$ has a strategy $\sigma_S\in \Sigma_S$
such that for any strategy $\sigma_{-S}\in \Sigma_{-S}$ of the complement,
$g(\sigma_S, \sigma_{-S})\in B$.\footnote{The notion of $\alpha$-effectivity is standard \citep[e.g.,][]{peleg02hbscw}.}
Define a simple game as the set of winning coalitions, 
where a coalition is \emph{winning} if it is $\alpha$-effective for all subsets of~$X$.
One can easily check that the winning coalitions for our $g$ are the coalitions containing at least $k-1$ players.

(ii)~A coalition $S$ is \emph{exactly effective for} a subset $B \subseteq X$ 
if $B=\{ g(\sigma_S, \sigma_{-S}) : \sigma_{-S}\in \Sigma_{-S}\}$
for some $\sigma_S\in \Sigma_S$.\footnote{
This notion of effectivity is proposed by \citet{kolpin90}.  
It is more informative than $\alpha$-effectivity.
Indeed, $S$ is $\alpha$-effective for~$B$ if and only if
there exists some $B'\subseteq B$ such that $S$ is exactly effective for~$B'$.
If a coalition $S$ is exactly effective (not just $\alpha$-effective) for a set~$B$ of at least two elements,
then the complement $S^c$ can realize \emph{every} (not just some) element in~$B$ by a suitable 
choice of strategies.
Intuitively, then, $S$ has the power to leave the others to choose from~$B$.
This notion is potentially more suitable for studying certain aspects of the theory of rights 
than $\alpha$-effectivity is,
since it describes a coalition's \emph{right to stay passive} more finely.
(\citet[Definition 11]{deb04} %
is an example of an application to the theory of rights.)
To show that $\alpha$-effectivity is inadequate, take, for example, 
``maximal freedom'' and the ``right to be completely passive'' by \citet{hees99}.
Van Hees resolves the liberal paradox by adopting either of these notions.
A necessary condition for maximal freedom is monotonicity with respect to
alternatives: if a coalition is effective for a set, it should be effective for a larger set.
A coalition is said to have the right to be completely passive if it is effective for the set~$X$
of all alternatives.  
Since $\alpha$-effectivity is monotonic with respect to alternatives
and since every coalition is $\alpha$-effective for~$X$,
$\alpha$-effectivity fails to capture the subtle, but important differences that these
notions can discriminate.
} %
Define a simple game as the set of winning coalitions, 
where a coalition is \emph{winning} if it is exactly effective for all subsets of~$X$.
Then, the winning coalitions for our $g$ are the coalitions made up of exactly $k-1$ players,
which confirms our claim.
In particular, the grand coalition $\{0, 1, \ldots, k-1\}$---while
 it is exactly effective for $\{0\}$ and $\{1\}$---is not 
exactly effective for $\{0, 1\}$, but a coalition made up of exactly $k-1$ players is.\end{remark}

Next, we move on to games without finite carriers.
We construct them using the notion of the \emph{product} of games.
By a \textbf{recursive function $f$ on a recursive set $T\subseteq N$} we mean 
a recursive function restricted to~$T$.

Let $(f_1, f_2)$ be a pair consisting of 
a one-to-one recursive function $f_1$ on a (not necessarily finite) recursive set $T\subseteq N$ and 
a one-to-one recursive function~$f_2$,
whose images partition the set of players: $f_1(T)\cap f_2(N)=\emptyset$
and $f_1(T)\cup f_2(N)=N$.
Note that $f_1^{-1}$ and $f_2^{-1}$ are 
recursive functions on recursive sets $f_1(T)$ and $f_2(N)$, respectively.\footnote{\label{image-of-rec}%
In general, if $f$ is a recursive function and $S$ is a recursive set, 
then the image $f(S)$ is recursively enumerable.
So $f_1(T)$ and $f_2(N)$ are recursively enumerable.
Since they complement each other on the set~$N$, they are in fact both recursive.}

We define the \textbf{disjoint image of coalitions $S_1\subseteq T$ and $S_2\subseteq N$ 
with respect to $(f_1, f_2)$} as the set
\[
S_1*S_2=f_1(S_1)\cup f_2(S_2),
\]
where $f_1(S_1)=\{f_1(i): i\in S_1\}$ and  $f_2(S_2)=\{f_2(i): i\in S_2\}$.

\begin{example}
When $T=N$, an easy example is given by $f_1: i\mapsto 2i$ and $f_2: i\mapsto 2i+1$.
In this case, $f_1(T)=2N:=\{2i: i\in N\}$, $f_2(N)=2N+1:=\{2i+1: i\in N\}$, and
$\{0, 2, 3\}*\{1, 2, 4\}=\{0, 4, 6, 3, 5, 9\}$.
When $T=\{0,1, \ldots, k-1\}$ for some $k \ge 1$, an easy example is given by $f_1: i\mapsto i$
 and $f_2: i\mapsto i+k$.  In this case, if $k=4$, we have
 $f_1(T)=T$, $f_2(N)=N\setminus T=\{4, 5, 6, \ldots\}$, and
 $\{0, 2, 3\}*\{1, 2, 4\}=\{0, 2, 3, 5, 6, 8\}$.\end{example}

\begin{lemma}\label{isREC}
Let $\REC$ be the class of (recursive) coalitions.  Then, 
\[
 \{S_1*S_2: \textup{$S_1\subseteq T$ and $S_2$ are coalitions}\}=\REC.
 \]
 \end{lemma}

\begin{proof}
($\subseteq$).  By an argument similar to that in footnote~\ref{image-of-rec},
$f_1(S_1)$ and $f_2(S_2)$ are recursive.
It follows that $f_1(S_1)\cup f_2(S_2)$ is recursive.

($\supseteq$).  Let $S$ be recursive.  Then
\begin{eqnarray*}
S & = & [S\cap f_1(T)]\cup [S\cap f_2(N)] \\
	& = & [f_1(f_1^{-1}( S\cap f_1(T) ))]\cup [f_2(f_2^{-1}( S\cap f_2(N) ))]
\end{eqnarray*}
\end{proof}

Let $\omega_1$ be a game with a carrier included in a set~$T$.
(This is without loss of generality since the grand coalition~$N$ is a carrier for any game.)
Let $\omega_2$ be a game.
We define the \textbf{product $\omega_1\otimes \omega_2$ of $\omega_1$ and $\omega_2$
with respect to $(f_1, f_2)$}
by the set
\[
 \omega_1\otimes \omega_2=\{f_1(S_1)\cup f_2(S_2): \textrm{$S_1\in \omega_1$ and $S_2\in\omega_2$}\}
\]
of the disjoint images of winning coalitions.\footnote{The notion of the \emph{product} of games is not new.
For example, \citet{shapley62} defines it for two games on disjoint subsets of players.}
By Lemma~\ref{isREC}, $\omega_1\otimes \omega_2$ is a simple game.
We have
$S_1*S_2\in \omega_1\otimes \omega_2$ if and only if $S_1\in \omega_1$ and  $S_2\in \omega_2$.

\begin{lemma} \label{product-comp}
If $\omega_1$ and $\omega_2$ are computable, then the product $\omega_1\otimes\omega_2$ is computable.
\end{lemma}

\begin{proof}
Let $e$ be a characteristic index for a coalition $S:=S_1*S_2=f_1(S_1)\cup f_2(S_2)$.
It suffices to show that given $e$, we can effectively obtain a characteristic index for $S_1$
 (and similarly for $S_2$).

Let $t$ be a characteristic index for~$f_1(T)$, a fixed recursive set.
Effectively obtain \citep[Corollary II.2.3]{soare87} from $e$ and $t$ a characteristic index $e'$ for
$f_1(S_1)=[f_1(S_1)\cup f_2(S_2)]\cap f_1(T)$.
Let $t'$ be an index for the recursive function
\[
\varphi_{t'}(i)= \left\{ \begin{array}{ll}
		f_1(i) & \mbox{if $i\in T$}\\
		f_2(0) & \mbox{otherwise.}
		\end{array}
	\right.
\]
We claim that $\varphi_{e'}\circ \varphi_{t'}$ is the characteristic function 
for the recursive set~$S_1$.
(\emph{Details}.  Suppose $i\in S_1$ first.  Then $i\in T$ and $f_1(i)\in f_1(S_1)$.
Hence  $\varphi_{e'}\circ \varphi_{t'}(i)= \varphi_{e'}(f_1(i))=1$.
Suppose $i\notin S_1$ next.  If $i\in T$, then $f_1(i)\in f_1(T)\setminus f_1(S_1)$.
Hence $\varphi_{e'}\circ \varphi_{t'}(i)= \varphi_{e'}(f_1(i))=0$.
If $i\notin T$, then $\varphi_{e'}\circ \varphi_{t'}(i)= \varphi_{e'}(f_2(0))=0$,
since $f_2(0)\notin f_1(S_1)$.)

By the Parameter Theorem \citep[I.3.5]{soare87},
there is a recursive function $g$ such that $\varphi_{g(e')}(i)=\varphi_{e'}\circ \varphi_{t'}(i)$,
implying that $g(e')$ is characteristic index for $S_1$ that can be obtained effectively.\end{proof}

It turns out that the construction based on the product is very useful for our purpose.

\begin{lemma} \label{product-monotonic}
$\omega_1$ and $\omega_2$ are monotonic if and only if the product $\omega_1\otimes\omega_2$ is monotonic.
\end{lemma}

\begin{proof}
By Lemma~\ref{isREC}, any coalition $\hat{S}$ 
can be written as $\hat{S}=\hat{S}_1*\hat{S}_2$ for some $\hat{S}_1\subseteq T$ and $\hat{S}_2$.  

($\Longrightarrow$).
Suppose $S_1*S_2\in \omega_1\otimes\omega_2$ and $S_1*S_2\subseteq S'_1*S'_2$.
Then, we have $S_1\in\omega_1$, $S_2\in \omega_2$, 
and $f_1(S_1)\cup f_2(S_2)\subseteq f_1(S'_1)\cup f_2(S'_2)$.
Noting that $f_1(S_1)\subseteq f_1(T)$, $f_1(S'_1)\subseteq f_1(T)$, 
$f_2(S_2)\subseteq f_2(N)$, $f_2(S'_2)\subseteq f_2(N)$, and $f_1(T)\cap f_2(N)=\emptyset$,
we have $f_1(S_1)\subseteq f_1(S'_1)$ and $f_2(S_2)\subseteq f_2(S'_2)$.
Hence $S_1\subseteq S'_1$ and $S_2\subseteq S'_2$.
Since $S_1\in\omega_1$ and $S_2\in \omega_2$, monotonicity implies that
$S'_1\in\omega_1$ and $S'_2\in \omega_2$.
That is, $S'_1*S'_2\in \omega_1\otimes \omega_2$.

($\Longleftarrow$).  We suppose that $\omega_1\otimes\omega_2$ is monotonic 
and show that $\omega_1$ is monotonic.
Suppose $S_1\in\omega_1$ and $S_1\subset S'_1$.
Choose any $S_2\in \omega_2$.
Then $S_1*S_2\in \omega_1\otimes\omega_2$.
By monotonicity, $S'_1*S_2\in \omega_1\otimes\omega_2$.
Hence $S'_1\in \omega_1$.\end{proof}

\begin{lemma} \label{product-proper}
If $\omega_1$ or $\omega_2$ is proper,  then the product $\omega_1\otimes\omega_2$ is proper.
\end{lemma}

\begin{proof}
First, we can show that $(S_1*S_2)^c=S_1^c*S_2^c$, where $S_1^c=T\setminus S_1$
and $S_2^c=N\setminus S_2$.
Indeed, $(S_1*S_2)^c=(f_1(S_1)\cup f_2(S_2))^c =
(f_1(S_1))^c\cap (f_2(S_2))^c = 
[f_1(T)\setminus f_1(S_1) \cup f_2(N)] \cap [f_1(T)\cup f_2(N)\setminus f_2(S_2)]
=f_1(T\setminus S_1) \cup f_2(N\setminus S_2)
=S_1^c*S_2^c$. 

Now suppose $S_1*S_2\in \omega_1 \otimes \omega_2$.
Then, $S_1\in \omega_1$ and $S_2\in \omega_2$.
Since $\omega_1$ or $\omega_2$ is proper, we have 
either $S_1^c\notin\omega_1$ or $S_2^c\notin \omega_2$.
It follows that $(S_1*S_2)^c=S_1^c*S_2^c\notin \omega_1 \otimes \omega_2$.\end{proof}

\begin{lemma} \label{product-nonstrong}
Suppose $\omega_1$ is nonstrong or $\omega_2$ is nonstrong
or both $\omega_1$ and $\omega_2$ have losing coalitions.
Then the product $\omega_1\otimes\omega_2$ is nonstrong.
\end{lemma}

\begin{proof}
We give a proof for the case where each game has a losing coalition: 
$S_1\notin\omega_1$ and $S_2^c\notin\omega_2$.
Then, $S_1*S_2\notin \omega_1 \otimes \omega_2$ and 
 $(S_1*S_2)^c=S_1^c*S_2^c\notin \omega_1 \otimes \omega_2$.\end{proof}

 \begin{lemma} \label{product-nonweak}
If $\omega_1$ and $\omega_2$ are nonweak, then the product $\omega_1\otimes\omega_2$ is nonweak.
Its Nakamura number is $\nu(\omega_1\otimes\omega_2)=\max\{ \nu(\omega_1), \nu(\omega_2)\}$.
\end{lemma}

\begin{proof}
If $\bigcap\omega_1=\bigcap\omega_2=\emptyset$, then 
$\bigcap (\omega_1\otimes\omega_2)
=\bigcap_{S_1*S_2\in \omega_1\otimes\omega_2} (S_1*S_2)
=\bigcap_{S_1\in \omega_1, S_2\in \omega_2} (f_1(S_1)\cup f_2(S_2))
=(\bigcap_{S_1\in \omega_1} f_1(S_1))\cup (\bigcap_{S_2\in \omega_2} f_2(S_2))$
[because $f_1(S_1)\cap f_2(S_2)=\emptyset$ for all $S_1$ and $S_2$]
$=f_1(\bigcap_{S_1\in \omega_1} S_1)\cup f_2(\bigcap_{S_2\in \omega_2} S_2)
=(\bigcap\omega_1) * (\bigcap\omega_2)=\emptyset$.
The proof for the Nakamura number is similar.\end{proof}

Propositions \ref{nakamura:type3} and \ref{nakamura:type11} have analogues for infinite games
(because of Lemma~\ref{mono-prop-nakamura} again, Proposition~\ref{nakamura:type3inf}
precludes $k=2$):

\begin{prop} \label{nakamura:type3inf}
For any $k\geq 3$,  there exists an infinite, computable, type~$3$ 
(i.e., monotonic proper nonstrong nonweak) game~$\omega$
with Nakamura number $\nu(\omega)=k$.\end{prop}

\begin{proof}
For $k\ge 3$, let $\omega_1$ be a finite, computable, type~$3$ game
with $\nu(\omega_1)=k$.  
(Such a game exists by Proposition~\ref{nakamura:type3}.)
Let $\omega_2$ be an infinite, computable, monotonic nonweak game (which need not be proper
or strong or nonstrong) with $\nu(\omega_2)\le 3$.  
(Such a game exists by Proposition~\ref{type1ex}.)
Lemmas \ref{product-comp}, \ref{product-monotonic}, \ref{product-proper}, \ref{product-nonstrong}, 
\ref{product-nonweak} imply that
the product $\omega_1\otimes\omega_2$ satisfies the conditions.\end{proof}

\begin{prop} \label{nakamura:type11inf}
For any $k\geq 2$,  there exists an infinite, computable, type~$11$ 
(i.e., nonmonotonic proper nonstrong nonweak) game~$\omega$
with Nakamura number $\nu(\omega)=k$.\end{prop}

\begin{proof}
For $k\ge 2$, let $\omega_1$ be a finite, computable, type~$11$ game
with $\nu(\omega_1)=k$.  
(Such a game exists by Proposition~\ref{nakamura:type11}.)
Let $\omega_2$ be an infinite, computable, nonproper game.
(Types 5, 7, 13, and 15 in \citet{kumabe-m07csg64} are examples.
Alternatively, just for obtaining the results for $k\ge 3$, we can let
$\omega_2$ be an infinite, computable, nonweak game with
$\nu(\omega_2)=3$, which exists by Proposition~\ref{type1ex}.)
Then the game is nonweak, 
with $\nu(\omega_2)=2$ (if $\emptyset\notin\omega_2$; Lemma~\ref{nonproper-nakamura2})
or $\nu(\omega_2)=1$ (otherwise).  
Lemmas \ref{product-comp}, \ref{product-monotonic}, \ref{product-proper}, \ref{product-nonstrong}, 
\ref{product-nonweak} imply that
the product $\omega_1\otimes\omega_2$ satisfies the conditions.\end{proof}

\appendix

\section{An Infinite, Computable, Type~1 Game}
\label{nice_game1}


\emph{We exhibit here an infinite, computable, type~1 (i.e., monotonic proper strong nonweak) simple game},
thus giving a proof to Proposition~\ref{type1ex}.
Though \citet{kumabe-m07csg64} give an example, the readers
not comfortable with recursion theory may find it too complicated.
In view of the fact that such a game is used in an important result (e.g., Proposition~\ref{nakamura:type3inf}) 
in this paper, it makes sense to give a simpler construction here.\footnote{%
One reason that the construction in \citet{kumabe-m07csg64} is complicated is that
they construct a \emph{family} of type~1 games $\omega[A]$, one for each recursive set~$A$,
while requiring \emph{additional conditions} that would later become useful for constructing other types of games.
In this appendix, we construct just one type~1 game, forgetting about the additional conditions.
Some aspects of the construction thus become more transparent in this construction.
The construction extends the one (not requiring the game to be of a particular type) in
\citet[Section~6.2]{kumabe-m08jme}.}

\bigskip

Our approach is to construct recursively enumerable (in fact, recursive) sets $T_0$ and $T_1$ 
of strings (of 0's and 1's) satisfying the conditions of Proposition~\ref{delta0det2}.
We first construct certain sets~$F_s$ of strings for $s\in\{0, 1, 2, \ldots\}$.
We then specify each of $T_0$ and $T_1$ using the sets~$F_s$, 
and construct a simple game $\omega$ according to Proposition~\ref{delta0det2}.
We conclude that the game is computable by checking (Lemmas~\ref{ex:nocarrier-rec}
and \ref{ex:nocarrier-det}) 
that $T_0$ and $T_1$ satisfy the conditions of Proposition~\ref{delta0det2}.
Finally, we show (Claims~\ref{comp:monotonic}, \ref{comp:properstrong}, and~\ref{comp:nonweakcarrier}) 
that the game satisfies the desired properties.

\medskip

\textbf{Notation}.  
Let $\alpha$ and $\beta$ be strings (of 0's and 1's).

Then $\alpha^c$ denotes the string of the length $|\alpha|$ such that $\alpha^c(i)=1-\alpha(i)$ for each $i<|\alpha|$; 
for example, $0110100100^c=1001011011$.
Occasionally, a string $\alpha$ is identified with the set $\{i: \alpha(i)=1\}$.
(Note however that $\alpha^c$ is occasionally identified with the set $\{i: \alpha(i)=0\}$,
but never with the set $\{i: \alpha(i)=1\}^c$.) 

$\alpha\beta$ (or $\alpha*\beta$) denotes the concatenation of $\alpha$ followed by $\beta$.

$\alpha\subseteq \beta$ means that $\alpha$ is an initial segment of $\beta$ ($\beta$ extends $\alpha$);
$\alpha \subseteq A$ means that $\alpha$ is an initial segment of a set~$A$.

Strings $\alpha$ and $\beta$ are \textbf{incompatible} if neither
$\alpha\subseteq \beta$ nor $\beta\subseteq\alpha$
(i.e., there is $k< \min\{|\alpha|,|\beta|\}$ such that $\alpha(k)\neq \beta(k)$).\enspace$\|$

\bigskip

Let $\{k_s\}_{s=0}^\infty$ be an effective listing (recursive enumeration) of the members of 
the recursively enumerable set $\{k : \varphi_k(k)\in \{0,1\}\}$, 
where $\varphi_k(\cdot)$ is the $k$th partial recursive function of one variable
(it is computed by the Turing program with code (G\"{o}del) number~$k$).
We can assume that $k_0\ge 2$ and all the elements $k_s$ are distinct.
Thus, 
\[ \CRec \subset \{k : \varphi_k(k)\in \{0,1\}\} = \{k_0, k_1, k_2, \ldots\}, \]
where $\CRec$ is the set of characteristic indices for recursive sets.

Let $l_{0}=k_0+1$, and for $s>0$, let $l_{s}=\max \{l_{s-1}, k_{s}+1\}$.
We have $l_s\geq l_{s-1}$ (that is, $\{l_s\}$ is an nondecreasing sequence of numbers) and 
$l_s>k_s$ for each $s$.  Note also that $l_s\geq l_{s-1}>k_{s-1}$, $l_s\geq l_{s-2}>k_{s-2}$, etc.\ 
imply that $l_s> k_s$, $k_{s-1}$, $k_{s-2}$, \ldots, $k_0$.

For each $s$, let $F_s$ be the finite set of strings $\alpha=\alpha(0)\alpha(1)\cdots\alpha(l_s-1)$
of length $l_s\ge 3$ such that  %
\begin{equation} \label{ex:nocarrier1}
\textrm{$\alpha(k_{s})=\varphi_{k_{s}}(k_{s})$ and for each $s'<s$,  $\alpha(k_{s'})=1-\varphi_{k_{s'}}(k_{s'})$.}
\end{equation}
Note that (\ref{ex:nocarrier1}) imposes no constraints on $\alpha(k)$ for $k\notin\{k_0,k_1,k_2, \ldots, k_s\}$, 
while it actually imposes constraints for all $k$ in the set, 
since $|\alpha|=l_s> k_s$, $k_{s-1}$, $k_{s-2}$, \ldots, $k_0$.
We observe that if $\alpha\in F_s\cap F_{s'}$, then $s=s'$.
 Let $F=\bigcup_{s}F_s$.

\begin{lemma} \label{Fincompatible}
Any two distinct elements $\alpha$ and $\beta$ in $F$ are \emph{incompatible}.  That is, 
we have neither $\alpha\subseteq \beta$ ($\alpha$ is an initial segment of~$\beta$)
nor $\beta\subseteq\alpha$ 
(i.e., there is $k< \min\{|\alpha|,|\beta|\}$ such that $\alpha(k)\neq \beta(k)$).
\end{lemma}

\begin{proof} 
Let $|\alpha|\leq |\beta|$, without loss of generality.  
 If $\alpha$ and $\beta$ have the same length, then the 
 conclusion follows since otherwise they become identical strings.
 If $l_s=|\alpha|< |\beta|=l_{s'}$, then $s<s'$ and by (\ref{ex:nocarrier1}),
 $\alpha(k_{s})=\varphi_{k_{s}}(k_{s})$ on the one hand, but 
$\beta(k_{s})=1-\varphi_{k_{s}}(k_{s})$ on the other hand.  So $\alpha(k_{s})\neq \beta(k_{s})$.\end{proof}

The game $\omega$ will be constructed from the sets $T_0$ and $T_1$ of strings defined as follows
($10=1*0$, $00=0*0$, and $11=1*1$ below):
\begin{eqnarray*}
\alpha\in T_0^0 & \iff & \exists s \, \textrm{[$\alpha\in F_s$, $\alpha\supseteq 10$,
 	and $\alpha(k_{s})(=\varphi_{k_{s}}(k_{s}))=0$]} \\
\alpha\in T_1^0 & \iff & \exists s \, \textrm{[$\alpha\in F_s$, $\alpha\supseteq 10$, 
	and $\alpha(k_{s})(=\varphi_{k_{s}}(k_{s}))=1$]}  \\
\alpha\in T_0 & \iff & \textrm{[$\alpha \in T_0^0$ or $\alpha^c\in T_1^0$ or $\alpha=00$]} \\
\alpha\in T_1 & \iff & \textrm{[$\alpha \in T_1^0$ or $\alpha^c\in T_0^0$ or $\alpha=11$]}.
\end{eqnarray*}
We observe that  the sets $T_0^0$, $T_1^0$, $T_0$, $T_1$ consist of strings whose lengths are at least~2,
$T_0^0\subset T_0$, $T_1^0\subset T_1$, $T_0\cap T_1=\emptyset$, and 
$\alpha\in T_0 \Leftrightarrow \alpha^c\in T_1$.

\emph{Define $\omega$ by $S\in \omega$ if and only if $S$ has an initial segment in $T_1$}.
Lemmas~\ref{ex:nocarrier-rec} 
and \ref{ex:nocarrier-det} establish computability of $\omega$ 
(as well as the assertion that $T_0$ consists of losing determining strings and
 $T_1$ consists of winning determining strings)
by way of Proposition~\ref{delta0det2}.

\begin{lemma} \label{ex:nocarrier-rec}
$T_0$ and $T_1$ are recursive.
\end{lemma}

\begin{proof}
We give an algorithm that can decide for each given string $\sigma$ with a length of at least~2 
whether it is in $T_0$ or in $T_1$ or neither.

If $\sigma\supseteq 00$, then $\sigma\notin T_0\cup T_1$ unless $\sigma=00\in T_0$.

If  $\sigma\supseteq 11$, then $\sigma\notin T_0\cup T_1$ unless $\sigma=11\in T_1$.

Suppose $\sigma\supseteq 10$.
In this case, $\sigma\in T_0\cup T_1$ iff $\sigma\in T_0^0\cup T_1^0$.
Generate $k_0$, $k_1$, $k_2$, \ldots, 
compute $l_0$, $l_1$, $l_2$, \ldots, and determine $F_0$, $F_1$, $F_2$, \ldots
until we find the least $s$ such that $l_s\geq |\sigma|$.

If $l_s > |\sigma|$, then $\sigma\notin F_s$.
Since $l_s$ is nondecreasing in $s$ and $F_s$ consists of strings of length~$l_s$, 
it follows that $\sigma\notin F$, implying $\sigma\notin T_0^0\cup T_1^0$, 
that is, $\sigma\notin T_0\cup T_1$.

If $l_s= |\sigma|$, then check whether $\sigma \in F_s$; this can be 
done since the values of $\varphi_{k_{s'}}(k_{s'})$ for $s'\leq s$ in (\ref{ex:nocarrier1})
are available and $F_s$ determined by time $s$.
If $\sigma\notin F_s$ and $l_{s+1}>l_s$, then $\sigma \notin T_0\cup T_1$ as before.
Otherwise check whether $\sigma \in F_{s+1}$.  
If $\sigma \notin F_{s+1}$ and $l_{s+2}>l_{s+1}=l_s$, then $\sigma \notin T_0\cup T_1$ as before.
Repeating this process, we either get $\sigma\in F_{s'}$ for some $s'$ or 
$\sigma\notin F_{s'}$ for all $s'\in \{s': l_{s'}=l_s\}$.
In the latter case, we have $\sigma \notin T_0\cup T_1$.
In the former case, if $\sigma(k_{s'})=\varphi_{k_{s'}}(k_{s'})=1$, 
then $\sigma \in T_1^0\subset T_1$ by the definitions of $T_1^0$ and $T_1$.
Otherwise $\sigma(k_{s'})=\varphi_{k_{s'}}(k_{s'})=0$, 
and we have $\sigma \in T_0^0\subset T_0$.

Suppose $\sigma\supseteq 01$.  Then $\sigma^c\supseteq 10$.
In this case the algorithm can decide whether $\sigma^c$ is in $T_0^0$ or in $T_1^0$ or neither.
If $\sigma^c\in T_0^0$, then $\sigma\in T_1$.  
If $\sigma^c\in T_1^0$, then $\sigma\in T_0$.
If $\sigma^c\notin T_0^0\cup T_1^0$, then $\sigma\notin T_0\cup T_1$.
\end{proof}

\begin{lemma} \label{t01incompatible}
Let $\alpha$, $\beta$ be distinct strings in $T_0\cup T_1$. Then $\alpha$ and $\beta$ are incompatible.
In particular, if $\alpha\in T_0$ and $\beta\in T_1$, then $\alpha$ and $\beta$ are incompatible.
\end{lemma}

\begin{proof}
Suppose $\alpha$ and $\beta$ are compatible.
Then there is a coalition~$S$ extending $\alpha$ and $\beta$.

If $\alpha\supseteq 00$, then $\beta\supseteq 00$.  But there is only one string in $T_0\cup T_1$
that extends $00$; namely, $00$.  So $\alpha=\beta=00$, contrary to the assumption that they are distinct.
The case where $\alpha\supseteq 11$ is similar.

If $\alpha\supseteq 10$, then $\beta\supseteq 10$.  
So we have $\alpha$, $\beta\in T_0^0\cup T_1^0$, which implies that $\alpha$, $\beta\in F$.
By Lemma~\ref{Fincompatible}, $S$ cannot extend both $\alpha$ and $\beta$, a contradiction.

If $\alpha\supseteq 01$, then $\beta\supseteq 01$.  
So we have $\alpha^c$, $\beta^c\in T_1^0\cup T_1^0$, which implies that $\alpha^c$, $\beta^c\in F$.
By Lemma~\ref{Fincompatible}, $S^c$ cannot extend both $\alpha^c$ and $\beta^c$, a contradiction.\end{proof}

\begin{lemma} \label{ex:nocarrier-string}
Let $\alpha\supseteq 10$ be a string of length $l_s$ such that $\alpha(k_s)=\varphi_{k_s}(k_s)$.
Then for some $t\le s$, there is a string $\beta\in F_t$ such that $10 \subseteq \beta \subseteq \alpha$.
\end{lemma}

\begin{proof}
We proceed by induction on $s$.
If $s=0$, we have $\beta=\alpha\in F_0$ (note that (\ref{ex:nocarrier1}) imposes no constraints on  $\alpha(0)$
and $\alpha(1)$); hence the lemma holds for $s=0$. 
Suppose the lemma holds for $s'<s$. If for some $s'<s$, $\alpha(k_{s'})=\varphi_{k_{s'}}(k_{s'})$,
then by the induction hypothesis, for some $t\leq s'$, the $l_{s'}$-initial segment $\alpha[l_{s'}]$ of $\alpha$
extends a string $\beta\in F_t$.  Hence the conclusion holds for $s$.
Otherwise,  we have for each $s'<s$,  $\alpha(k_{s'})=1-\varphi_{k_{s'}}(k_{s'})$. 
Then by (\ref{ex:nocarrier1}), $\alpha\in F_s$.  Letting $\beta=\alpha$ gives the conclusion.\end{proof}

\begin{lemma} \label{ex:nocarrier-det}
Any coalition $S\in\REC$ has an initial segment in $T_0$ or in $T_1$, but not both.
\end{lemma}

\begin{proof}
We show that $S$~has an initial segment in $T_0\cup T_1$. 
Lemma~\ref{t01incompatible} implies that
 $S$ does not have initial segments in both $T_0$ and $T_1$.
(The assertion following ``In particular'' in Lemma~\ref{t01incompatible} is sufficient for this,
but we can actually show the stronger statement that $S$ has exactly one initial segment in $T_0\cup T_1$.)

The conclusion is obvious if $S\supseteq 00$ or $S\supseteq 11$.

If $S\supseteq 10$, suppose $\varphi_k$ is the characteristic function for~$S$. 
Then $k\in\{k_0,k_1,k_2, \ldots\}$ since this set contains the set $\CRec$ of characteristic indices.
So $k=k_s$ for some~$s$.  Consider the initial segment $S[l_s]:=S\cap l_s=\varphi_{k_s}[l_s] \supseteq 10$.  
By Lemma~\ref{ex:nocarrier-string}, for some $t\le s$, there is a string $\beta\in F_t$ such that 
$10 \subseteq \beta \subseteq S[l_s]$.
The conclusion follows since $\beta$  is an initial segment of $S$ and 
$\beta\in T_0^0\cup T_1^0\subset T_0\cup T_1$.

If $S\supseteq 01$, then $S^c\supseteq 10$ has an initial segment $\beta\in T_0^0\cup T_1^0$
by the argument above.  So, $S$ has the initial segment $\beta^c\in T_1\cup T_0$.\end{proof}

Next, we show that the game~$\omega$ has the desired properties.  Before showing monotonicity, 
we need the following lemma.  For strings $\alpha$ and $\beta$ with $|\alpha|\le |\beta|$, 
we say \emph{$\beta$ properly contains~$\alpha$}
if for each $k<|\alpha|$, $\alpha(k)\leq\beta(k)$ and for some $k'<|\alpha|$,  $\alpha(k')<\beta(k')$; 
we say \emph{$\beta$~is properly contained by~$\alpha$}
 if for each $k<|\alpha|$, $\beta(k)\le \alpha(k)$ and for some $k'<|\alpha|$, $\beta(k')<\alpha(k')$.

\begin{lemma}\label{conainingstrings}
Let $\alpha$ and $\beta$ be strings such that $|\alpha|\le |\beta|$.
\textup{(i)}~If $\alpha\in T_1$ and  $\beta$ properly contains~$\alpha$, then $\beta$ extends a string in~$T_1$.
\textup{(ii)}~If $\alpha\in T_0$ and  $\beta$ is properly contained by~$\alpha$, then $\beta$ extends a string in~$T_0$.
\end{lemma}

\begin{proof}
(i)~Suppose $\alpha\in T_1$ and $\beta$ properly contains~$\alpha$.
We have $\alpha=11$ or $\alpha\in T_1^0$ or $\alpha^c\in T_0^0$.

If $\alpha=11$, no $\beta$ properly contains~$\alpha$.

Suppose $\alpha\in T_1^0$.   Then $\alpha\in F_s$ for some~$s$.
Since $\beta$ properly contains~$\alpha\supseteq 10$, we have
$\beta\supseteq 11$ or $\beta\supseteq 10$.
If $\beta\supseteq 11$, the conclusion follows since $11\in T_1$.
Otherwise, $\beta\supseteq 10$; 
choose the least $s'\leq s$ such that $\beta(k_{s'})=\varphi_{k_{s'}}(k_{s'})=1$.
(Such an $s'$ exists since $\alpha\in T_1^0$ implies $\alpha(k_{s})=\varphi_{k_{s}}(k_{s})=1$.  
Note that $k_{s'}< l_{s'}\le l_s=|\alpha|$.)
Then for each $t<s'$, we have $\beta(k_t)=1-\varphi_{k_t}(k_t)$.  
(\emph{Details}.  By the choice of~$s'$, for each $t<s'$, either (a)~$\beta(k_t)=\varphi_{k_t}(k_t)=0$ or
(b)~$\beta(k_t)\neq \varphi_t(k_t)$.  
Suppose~(a) for some $t<s'$. 
Since $\alpha\in F_s$, we have for each $t<s$, $\alpha(k_t)=1-\varphi_{k_t}(k_t)$ by (\ref{ex:nocarrier1}).  Then we have $\beta(k_t)=0$ and $\alpha(k_t)=1$, contradicting the assumption that
$\beta$ properly contains~$\alpha$.)
The conclusion follows since the initial segment $\beta[l_{s'}]$ is in~$T_1^0$.

Suppose $\alpha^c\in T_0^0$.   Then $\alpha^c\in F_s$ for some~$s$.
Since $\beta^c$ is properly contained in~$\alpha^c\supseteq 10$, we have
$\beta^c\supseteq 00$ or $\beta^c\supseteq 10$.
If $\beta^c\supseteq 00$, the conclusion follows since $\beta\supseteq 11\in T_1$.
Otherwise, $\beta^c\supseteq 10$;
Choose the least $s'\leq s$ such that $\beta^c(k_{s'})=\varphi_{k_{s'}}(k_{s'})=0$.
Then for each $t<s'$, we have $\beta^c(k_t)=1-\varphi_{k_t}(k_t)$ as before.
Therefore, the initial segment $\beta^c[l_{s'}]$ is in~$T_0^0$.
The conclusion follows since $\beta[l_{s'}]\in T_1$.

(ii)~Suppose $\alpha\in T_0$ and  $\beta$ is properly contained by~$\alpha$.
Then $\alpha^c\in T_1$ and  $\beta^c$ properly contains~$\alpha^c$.
Assertion~(i) then implies that $\beta^c$ extends a string $\beta^c[l_{s'}]$ in~$T_1$.
Therefore, $\beta$ extends the string $\beta[l_{s'}]$ in~$T_0$.\end{proof}

\begin{claim}\label{comp:monotonic}
The game~$\omega$ is monotonic.
 \end{claim}

\begin{proof}
Suppose $A\in \omega$ and $B\supseteq A$.  
By the definition of $\omega$, $A$ has an initial segment $\alpha$ in~$T_1$.
If $B$ extends~$\alpha$, then clearly $B\in \omega$.
Otherwise the $|\alpha|$-initial segment $\beta=B[|\alpha|]$ of $B$ properly contains~$\alpha$.
By Lemma~\ref{conainingstrings}, $\beta$ extends a string in~$T_1$. 
Hence $B$ has an initial segment in~$T_1$, implying that $B\in \omega$.\end{proof}
 
\begin{claim}\label{comp:properstrong}
The game~$\omega$ is proper and strong.
\end{claim}

\begin{proof}
It suffices to show that $S^c\in\omega \Leftrightarrow S\notin\omega$.
From the observations that
$T_0$ and $T_1$ consist of determining strings and
that $\alpha^c\in T_0 \Leftrightarrow \alpha \in T_1$, we have
\begin{eqnarray*}
S^c \in\omega & \iff & \textrm{$S^c$ has an initial segment in~$T_1$} \\
						& \iff & \textrm{$S$ has an initial segment in~$T_0$} \\
						& \iff & S\notin\omega.
\end{eqnarray*}\end{proof}

\begin{claim} \label{comp:nonweakcarrier}
The game~$\omega$ is nonweak and does not have a finite carrier.
\end{claim}

\begin{proof}
To show that the game does not have a finite carrier, we will construct a set~$A$ such that 
for infinitely many~$l$,
the $l$-initial segment $A[l]$ has an extension (as a string) that is winning 
and for infinitely many $l'$, $A[l']$ has an extension that is losing.  
This implies that $A[l]$ is not a carrier of~$\omega$ for any such~$l$.
So no subset of $A[l]$ is a carrier.  Since there are arbitrarily large such~$l$, 
this proves that $\omega$ has no finite carrier.

Let $A\supseteq 10$ be a set such that for each $k_t$, $A(k_t)=1-\varphi_{k_t}(k_t)$. 
For any $s'>0$ and $i \in \{0,1\}$, there is an $s>s'$ such that $k_s> l_{s'}$ and $\varphi_{k_s}(k_s)=i$.

For a temporarily chosen $s'$, fix $i$ and fix such $s$.  Then choose the greatest $s'$ satisfying these conditions.
Since $l_s>k_s>l_{s'}$, there is a string $\alpha$ of length~$l_s$ extending $A[l_{s'}]$ 
such that $\alpha\in F_s$.  
Since $\alpha\supseteq 10$ and $\alpha(k_s)=\varphi_{k_s}(k_s)=i$, we have $\alpha\in T_i^0$.  

There are infinitely many such $s$, so there are infinitely many such $s'$. 
It follows that for infinitely many $l_{s'}$, the initial segment $A[l_{s'}]$ is a substring 
of some  string~$\alpha$ in $T_1$, 
and for infinitely many $l_{s'}$, $A[l_{s'}]$ is a substring of some (losing) string~$\alpha$ in~$T_0$.

To show nonweakness, we give three (winning) coalitions in $T_1$ whose intersection is empty.
First, $10$ (in fact any initial segment of the coalition $A \supseteq 10$)
has extensions $\alpha$ in $T_1$ and $\beta$ in~$T_0$ by the argument above.
So $01$ has the extension~$\beta^c$ in~$T_1$.
Clearly, the intersection of the winning coalitions $11\in  T_1$, $\alpha\supseteq 10$, and 
$\beta^c\supseteq 01$ is empty.\end{proof}


\begin{thebibliography}{}

\bibitem[Andjiga and Mbih, 2000]{andjiga-m00}
Andjiga, N.~G. and Mbih, B. (2000).
\newblock A note on the core of voting games.
\newblock {\em Journal of Mathematical Economics}, 33:367--372.

\bibitem[Arrow, 1963]{arrow63}
Arrow, K.~J. (1963).
\newblock {\em Social Choice and Individual Values}.
\newblock Yale University Press, New Haven, 2nd edition.

\bibitem[Austen-Smith and Banks, 1999]{austensmith-b99}
Austen-Smith, D. and Banks, J.~S. (1999).
\newblock {\em Positive Political Theory I: Collective Preference}.
\newblock University of Michigan Press, Ann Arbor.

\bibitem[Banks, 1995]{banks95}
Banks, J.~S. (1995).
\newblock Acyclic social choice from finite sets.
\newblock {\em Social Choice and Welfare}, 12:293--310.

\bibitem[Bartholdi et~al., 1989a]{bartholdi-tt89vs}
Bartholdi, III, J., Tovey, C.~A., and Trick, M.~A. (1989a).
\newblock Voting schemes for which it can be difficult to tell who won the
  election.
\newblock {\em Social Choice and Welfare}, 6:157--165.

\bibitem[Bartholdi et~al., 1989b]{bartholdi-tt89cd}
Bartholdi, III, J.~J., Tovey, C.~A., and Trick, M.~A. (1989b).
\newblock The computational difficulty of manipulating an election.
\newblock {\em Social Choice and Welfare}, 6:227--241.

\bibitem[Deb, 2004]{deb04}
Deb, R. (2004).
\newblock Rights as alternative game forms.
\newblock {\em Social Choice and Welfare}, 22:83--111.

\bibitem[Kelly, 1988]{kelly88}
Kelly, J.~S. (1988).
\newblock Social choice and computational complexity.
\newblock {\em Journal of Mathematical Economics}, 17:1--8.

\bibitem[Kolpin, 1990]{kolpin90}
Kolpin, V. (1990).
\newblock Equivalent game forms and coalitional power.
\newblock {\em Mathematical Social Sciences}, 20:239--249.

\bibitem[Kumabe and Mihara, 2007]{kumabe-m07csg64}
Kumabe, M. and Mihara, H.~R. (2007).
\newblock 
{Computability of simple games: A complete investigation of the sixty-four possibilities}.
\newblock MPRA Paper 4405, Munich University Library.
[Notes added July 2011:  Journal of Mathematical Economics, 47:150--158, 2011]


\bibitem[Kumabe and Mihara, 2008]{kumabe-m08jme}
Kumabe, M. and Mihara, H.~R. (2008).
\newblock 
{Computability of simple games: A characterization and application to the core}.
\newblock {\em Journal of Mathematical Economics}, 44: 348-366.

\bibitem[Lewis, 1988]{lewis88}
Lewis, A.~A. (1988).
\newblock An infinite version of {A}rrow's {T}heorem in the effective setting.
\newblock {\em Mathematical Social Sciences}, 16:41--48.

\bibitem[Mihara, 1997]{mihara97et}
Mihara, H.~R. (1997).
\newblock {A}rrow's {T}heorem and {T}uring computability.
\newblock {\em Economic Theory}, 10:257--76.

\bibitem[Mihara, 1999]{mihara99jme}
Mihara, H.~R. (1999).
\newblock Arrow's theorem, countably many agents, and more visible invisible
  dictators.
\newblock {\em Journal of Mathematical Economics}, 32:267--287.

\bibitem[Mihara, 2004]{mihara04mss}
Mihara, H.~R. (2004).
\newblock Nonanonymity and sensitivity of computable simple games.
\newblock {\em Mathematical Social Sciences}, 48:329--341.

\bibitem[Nakamura, 1979]{nakamura79}
Nakamura, K. (1979).
\newblock The vetoers in a simple game with ordinal preferences.
\newblock {\em International Journal of Game Theory}, 8:55--61.

\bibitem[Odifreddi, 1992]{odifreddi92}
Odifreddi, P. (1992).
\newblock {\em Classical Recursion Theory: The Theory of Functions and Sets of
  Natural Numbers}.
\newblock Elsevier, Amsterdam.

\bibitem[Peleg, 2002]{peleg02hbscw}
Peleg, B. (2002).
\newblock Game-theoretic analysis of voting in committees.
\newblock In Arrow, K.~J., Sen, A.~K., and Suzumura, K., editors, {\em Handbook
  of Social Choice and Welfare}, volume~1, chapter~8, pages 395--423. Elsevier,
  Amsterdam.

\bibitem[Richter and Wong, 1999]{richter-w99jme}
Richter, M.~K. and Wong, K.-C. (1999).
\newblock Computable preference and utility.
\newblock {\em Journal of Mathematical Economics}, 32:339--354.

\bibitem[Shapley, 1962]{shapley62}
Shapley, L.~S. (1962).
\newblock Simple games: An outline of the descriptive theory.
\newblock {\em Behavioral Science}, 7:59--66.

\bibitem[Soare, 1987]{soare87}
Soare, R.~I. (1987).
\newblock {\em Recursively Enumerable Sets and Degrees: A Study of Computable
  Functions and Computably Generated Sets}.
\newblock Springer-Verlag, Berlin.

\bibitem[Tanaka, 2007]{tanaka07}
Tanaka, Y. (2007).
\newblock Type two computability of social choice functions and the
  {G}ibbard-{S}atterthwaite theorem in an infinite society.
\newblock {\em Applied Mathematics and Computation}, 192:168--174.

\bibitem[Truchon, 1995]{truchon95}
Truchon, M. (1995).
\newblock Voting games and acyclic collective choice rules.
\newblock {\em Mathematical Social Sciences}, 29:165--179.

\bibitem[van Hees, 1999]{hees99}
van Hees, M. (1999).
\newblock Liberalism, efficiency, and stability: Some possibility results.
\newblock {\em Journal of Economic Theory}, 88:294--309.

\bibitem[Weber, 1994]{weber94}
Weber, R.~J. (1994).
\newblock Games in coalitional form.
\newblock In Aumann, R.~J. and Hart, S., editors, {\em Handbook of Game
  Theory}, volume~2, chapter~36, pages 1285--1303. Elsevier, Amsterdam.

\end{thebibliography}






\pagebreak
\pagenumbering{arabic}

\section{Type 9 and Type 13 Games}
\label{sect:9and13}

In this attachment,
we modify the examples of a type~9 game and a type~13 game in \citet{kumabe-m07csg64}
 so that an empty coalition is losing.
To do that, modify the infinite, computable, type~1 game $\omega[A]$ in that paper as follows
((2.i) and (3) refer to certain requirements in that paper):

\begin{enumerate}
\item[9.] An infinite, computable, type~9 (nonmonotonic proper strong nonweak) game.
In the construction of $\omega[A]$, replace (2.i) by 
\begin{enumerate}
\item [(2*.i)] for each p-string $\alpha'\neq 10$ that is a proper substring of $\alpha$, 
	if $s=0$ or $|\alpha'|\geq l_{s-1}$, then enumerate $\alpha'*11$ in $T_1$ and $\alpha'*00$ in~$T_0$;
	furthermore, enumerate $1011$ and $1000$ in $T_0$. 
\end{enumerate}
By (3) of the construction of~$\omega[A]$, $0100, 0111\in T_1$.
(In other words, the game is constructed from the sets $T_0:=T'_0\cup \{1011\}\setminus\{0100\}$  
and $T_1:=T'_1\cup \{0100\}\setminus \{1011\}$, 
where $T'_0$ and $T'_1$ are $T_0$ and $T_1$ in the original construction of $\omega[A]$ renamed.
Note that $1011\in T'_1$, $1000\in T'_0$, $0100\in T'_0$, and $0111\in T'_1$.)
Letting $\alpha'=\emptyset$ in (2*.i), we have $00\in T_0$; so $\emptyset$ is losing.
Since either $\alpha'=1010$ or $1001$ is a p-string satisfying the condition in~(2*.i),
either $101011\in T_1$ or $100111\in T_1$.  
Then by (3), either $010100\in T_0$ or $011000\in T_0$.
So the game is nonmonotonic, since $0100$ is winning.
It is also nonweak since $0100$ is winning and either $101011$ or $100111$ is winning.  
For the remaining properties, the proofs  are similar to the proofs for~$\omega[A]$. 

\item [13.] An infinite, computable, type~13 (nonmonotonic nonproper strong nonweak) game.
In the construction of $\omega[A]$, replace (2.i) and (3) by
\begin{enumerate}
\item [(2*.i)] for each p-string $\alpha'\neq 10$ that is a proper substring of~$\alpha$, 
	if $s=0$ or $|\alpha'|\geq l_{s-1}$, then enumerate $\alpha'*11$ in $T_1$ and $\alpha'*00$ in~$T_0$;
	furthermore, enumerate $1011$ and $0100$ in $T_1$ and $1000$ in $T_0$;
	\item [(3*)] if a string $\beta\notin\{1011,0100\}$ is enumerated in $T_1$ (or in $T_0$) above, 
then enumerate $\beta^c$ in $T_0$ (or in $T_1$, respectively).
\end{enumerate}
By (3*), $0111\in T_1$. 
(In other words, the game is constructed from the sets $T_0:=T'_0\setminus \{0100\}$  
and $T_1:=T'_1\cup \{0100\}$, 
where $T'_0$ and $T'_1$ are $T_0$ and $T_1$ in the original construction of $\omega[A]$ renamed.)

By an argument similar to that for type~9, $\emptyset$ is losing and the game is nonmonotonic
(either $010100\in T_0$ or $011000\in T_0$, while $0100$ is winning).
It is nonproper since the $0100$ and $1011$ are winning determining.
It is strong since its subset $\omega[A]$ is strong.
It is nonweak by Lemma~\ref{weakisproper} since it is nonproper.
The proofs of computability and nonexistence of a finite carrier are similar to the proofs for~$\omega[A]$.
\end{enumerate}

 \end{document}